\documentclass[journal,twoside,web]{ieeecolor}

\usepackage{generic}
\usepackage{cite}
\usepackage{amsmath,amssymb,amsfonts}
\usepackage{algorithmic}
\usepackage{graphicx}
\usepackage{textcomp}

\usepackage{graphicx}
\usepackage{subfigure}
\usepackage{epsfig} 
\usepackage{cite}
\usepackage{lcsys}

\usepackage{algorithm}
\usepackage{mathtools}
\usepackage{xcolor}
\usepackage{hyperref}
\hypersetup{hidelinks=true}

\newtheorem{definition}{Definition}
\newtheorem{theorem}{Theorem}
\newtheorem{lemma}{Lemma}

\newtheorem{remark}{Remark}

\newtheorem{example}{Example}

\begin{document}

\def\BibTeX{{\rm B\kern-.05em{\sc i\kern-.025em b}\kern-.08em
    T\kern-.1667em\lower.7ex\hbox{E}\kern-.125emX}}
\markboth{\journalname, VOL. XX, NO. XX, XXXX 2017}
{Author \MakeLowercase{\textit{et al.}}: Preparation of Papers for IEEE Control Systems Letters (August 2022)}

\title{Exact Model-Free Policy Iteration for Co-safe LTL Planning}

\author{
Zetong Xuan and Yu Wang, \IEEEmembership{Senior Member, IEEE}
\thanks{This work was supported in part by the NSF CAREER Award No. 2541603 and ARO under Grant No. W911NF-25-1-0243.}
\thanks{The authors are with the Department of Mechanical and Aerospace Engineering, University of Florida, Gainesville, FL 32611 USA (e-mail: z.xuan@ufl.edu; yuwang1@ufl.edu).}%
\thanks{This is the accepted version of a manuscript accepted for publication in the \emph{IEEE Control Systems Letters}. The version of record is available at \url{https://doi.org/10.1109/LCSYS.2026.3702192}.}
\thanks{\copyright~2026 IEEE. Personal use of this material is permitted. Permission from IEEE must be obtained for all other uses, in any current or future media, including reprinting/republishing this material for advertising or promotional purposes, creating new collective works, for resale or redistribution to servers or lists, or reuse of any copyrighted component of this work in other works.}
}

\maketitle
\thispagestyle{empty}

\begin{abstract}
This work studies model-free reinforcement learning for co-safe linear temporal logic (sc-LTL) objectives in finite Markov decision processes, which can be reduced to maximal reachability objectives via the standard product construction. For this problem, direct sample-based bootstrap methods (e.g., TD or Q-learning) may fail to converge to optimal policies due to the noncontractive nature and nonuniqueness of solutions to the Bellman equation. We develop a new two-step model-free reinforcement learning method that first uses a discounted surrogate to identify a clamp set that resolves this nonuniqueness, and then applies undiscounted policy evaluation and greedy policy improvement with guarantees of finding an optimal solution. We prove almost-sure convergence of the policy evaluation step and finite termination of the policy iteration algorithm at an optimal policy. These theoretical results are validated through numerical experiments on a stochastic grid world.
\end{abstract}

\begin{IEEEkeywords}
Markov decision processes, Reachability analysis, Reinforcement learning, Model checking.
\end{IEEEkeywords}

n\vspace{-10pt}
\section{Introduction}
\vspace{-5pt}
\IEEEPARstart{M}{odel}-free reinforcement learning (RL) for co-safe linear temporal logic (sc\mbox{-}LTL) objectives~\cite{kupferman2001model} is important for autonomous systems that accomplish complex rule-based tasks under unknown stochastic dynamics. Sc\mbox{-}LTL is the co-safe fragment of linear temporal logic (LTL)~\cite{baier2008principles}, whose satisfaction can be determined by a finite execution prefix, and is therefore naturally suited to tasks such as sequential completion, ordered visitation, and event-triggered termination~\cite{kressgazit2009temporallogicbased,smith2011optimal,lacerda2014optimal}. When the model of the underlying Markov decision process (MDP) is unavailable, one must rely on sampled interactions to synthesize policies. In this setting, model-free methods are especially attractive because they are typically more scalable than model-based methods that rely on explicit model construction~\cite{fu2014pac}, which faces a memory bottleneck at scale.

Existing model-free RL methods for general LTL or LTL$_f$ 
introduce an unnecessary layer of indirection when applied to 
sc\mbox{-}LTL, incurring a gap between the learned policy and 
the optimal policy for the satisfaction probability. The 
issue for general LTL is that it is designed to express 
$\omega$-regular behavior, whereas for sc\mbox{-}LTL, the 
desired objective already becomes a reachability probability 
after the standard product 
construction~\cite{baier2008principles,lacerda2014optimal,
ding2014optimal}. These methods optimize the expected return 
of a surrogate reward rather than the satisfaction probability 
itself. The most common approach uses discounted rewards so 
that the resulting expected return approximates the 
satisfaction probability~\cite{hasanbeig2018logicallyconstraineda,
hasanbeig2023certified,bozkurt2020control,bozkurt2021modelfreea,
shao2023sample,kantaros2024sample}, but the discounted return 
may differ substantially from the true satisfaction 
probability for practical discount 
factors~\cite{hahn2020faithful,voloshin2023eventual}. 
Setting $\gamma = 1$ to recover the exact satisfaction 
probability is not applicable, since the Bellman operator is 
no longer a contraction and these algorithms lose their 
convergence guarantees. Average-reward 
alternatives~\cite{le2024reinforcement} avoid this gap but 
enlarge the state space with reward-machine states. 
Although sc\mbox{-}LTL coincides with LTL$_f$ for 
reachability on finite MDPs~\cite{degiacomo2013ltlf}, to 
the best of our knowledge, no model-free LTL$_f$ work 
directly computes an optimal policy for the satisfaction 
probability. Existing methods rely on analogous discounted 
weighted~\cite{degiacomo2019restraining} or 
indicator-style~\cite{toroicarte2018lpopl} rewards, incurring 
the same gap.

A direct model-free treatment of the reachability problem induced by sc\mbox{-}LTL is nevertheless nontrivial. The Bellman equation for reachability probability may admit multiple solutions in the absence of discounting~\cite{brazdil2014verification}. 
When recurrent structure is present, other fixed-points may arise. 
Consequently, sample-based methods driven by bootstrapping, such as temporal-difference (TD) updates, may not converge to the true reachability probability.
Existing sample-based methods for reachability also have important limitations. Some require additional assumptions or side information, such as knowledge of the topology of the underlying graph~\cite{ashok2019paca}, expected conditional distance information~\cite{svoboda2024reinforcement}, or model knowledge for identifying maximal end components in tree search~\cite{ashok2018monte}.

This work develops an exact model-free RL method for sc\mbox{-}LTL objectives and establishes convergence guarantees that do not rely on surrogate approximation or graph search.
The key idea is to identify and remove the recurrent structure responsible for the nonuniqueness of solutions, namely the clamp set.
We first use a discounted formulation whose value function reveals whether the target can ever be reached.
A positive value indicates that the target can be reached with positive probability, whereas a zero value identifies recurrent structure away from the target.
After eliminating this source of nonuniqueness, the Bellman equation becomes well\mbox{-}posed on the remaining states.
The true reachability probabilities can then be computed and used for greedy policy improvement.
Unlike \cite{brazdil2014verification}, which identifies end components by thresholding the repeated occurrence of state-action pairs along sampled paths, our method identifies the relevant recurrent structure through value function estimation, thereby utilizing bootstrapping and avoiding reliance on prescribed sampled-path lengths.

Our work is different from existing sample-based methods for sc\mbox{-}LTL that consider modified formulations rather than the exact reachability problem induced by sc\mbox{-}LTL.
Specifically, \cite{li2019topological} solves a discounted surrogate problem via topology-guided approximate dynamic programming, while \cite{li2017samplingbased} performs search over a parameterized policy class. Relatedly, \cite{cohen2023temporal} studies sc\mbox{-}LTL control for hybrid systems via a reach-avoid decomposition and approximately solves the resulting subproblems in a model-based RL method.
\vspace{-5pt}
\section{Preliminaries}\label{sec:prelim}
\vspace{-5pt}
This section introduces how an sc\mbox{-}LTL planning problem can be reduced to a maximal reachability problem.
By the standard automata-based product construction, an sc\mbox{-}LTL objective can be translated into a reachability objective
on a finite product MDP~\cite{baier2008principles,lacerda2014optimal}.
To streamline the presentation, we work directly with finite MDPs with reachability objectives and recall the Bellman equations for reachability.

\begin{definition}
\label{def:mdp}
A finite Markov decision process with reachability objective is a tuple
$\mathcal{M} = (S, A, T, \beta, Z)$ where
(i) $S$ is a finite set of states and $\beta\in \Delta(S)$ is the initial-state distribution, where $\Delta(S)$ denotes the set of probability distributions on $S$;
(ii) $A$ is a finite set of actions, and $A(s)$ denotes the set of admissible actions at $s\in S$;
(iii) \(T(\cdot|s,a)\) is the transition distribution for each admissible state-action pair \((s,a)\); 
(iv) $Z\subseteq S$ is the target set to be reached.
\end{definition}

A path of $\mathcal{M}$ is an infinite state sequence $\sigma=s_0s_1s_2\cdots$ such that for all $t\ge 0$,
there exists $a_t\in A(s_t)$ with $T(s_{t+1}|s_t,a_t)>0$.
We say that $\sigma$ satisfies the reachability objective if there exists $t\in\mathbb N$ such that $s_t\in Z$.

\begin{definition}
\label{def:policy}
A deterministic stationary policy is a function $\pi:S\to A$ such that $\pi(s)\in A(s)$ for all $s\in S$.
Given $\pi$, the induced Markov chain has transition kernel $T^\pi(s'|s):=T(s'|s,\pi(s))$ for all $s,s'\in S$.
\end{definition}

Let \(\mathbb P_s^\pi\) and \(\mathbb E_s^\pi\) denote the probability measure and expectation, respectively, under policy \(\pi\) with initial state \(s_0=s\).

The reachability formulation studied here captures a broad class of sc\mbox{-}LTL planning problems. An sc\mbox{-}LTL objective can be translated into a reachability objective by the standard product construction, which builds a product MDP from the original MDP and a deterministic finite automaton (DFA) generated from the sc\mbox{-}LTL formula~\cite{baier2008principles,lacerda2014optimal}. The automaton state encodes the specification progress and is included in the product state. Thus, a stationary policy on the product MDP corresponds to a finite-memory policy on the original MDP, and it suffices to search for an optimal deterministic stationary policy on the product MDP~\cite{baier2008principles}. On the product MDP, reaching the target set is equivalent to satisfying the sc\mbox{-}LTL objective.

For any subset $C\subseteq S$, define the hitting time $\tau_C := \inf\{t\ge 0 : s_t\in C\}$, with the convention $\inf\emptyset=\infty$.
For a policy $\pi$, the value function $V^\pi:S\to[0,1]$ is defined by
\begin{align}
V^\pi(s) := \mathbb P_s^\pi(\tau_Z<\infty).
\label{eq:Vpi_def}
\end{align}
It satisfies the Bellman equation
\begin{align}
V^\pi(s)
=
\begin{cases}
1, & s\in Z,\\
\sum_{s'\in S} T^\pi(s'|s)\,V^\pi(s'), & s\notin Z.
\end{cases}
\label{eq:BE_pi}
\end{align}
The associated state-action value function is defined as $Q^\pi(s,a):=\sum_{s'\in S} T(s'|s,a)\,V^\pi(s')$ for $s\in S$, $a\in A(s)$.
Equivalently, for $s\notin Z$, one has $V^\pi(s)=Q^\pi(s,\pi(s))$.

The maximal reachability problem seeks an optimal policy $\pi^\star$ and the corresponding optimal value function, 
\begin{align}
\pi^\star \in \arg\max\nolimits_\pi \mathbb E_{s_0\sim\beta}[V^\pi(s_0)], \quad V^\star(s) := \max\nolimits_\pi V^\pi(s). \notag
\end{align}
Moreover, $V^\star$ satisfies the Bellman optimality equation
\begin{align}
V^\star(s)
=
\begin{cases}
1, & s\in Z,\\[2pt]
\max\limits_{a\in A(s)} \sum\limits_{s'\in S} T(s'|s,a)\,V^\star(s'), & s\notin Z.
\end{cases}
\label{eq:BOE}
\end{align}

It is well known that \eqref{eq:BE_pi} and \eqref{eq:BOE} may admit multiple fixed-point solutions.
Meanwhile $V^\pi$ and $V^\star$ correspond to the \emph{least fixed-point} solutions of
\eqref{eq:BE_pi} and \eqref{eq:BOE}, respectively~\cite{baier2008principles}.
\vspace{-5pt}
\section{Motivation}\label{sec:motivation}
\vspace{-5pt}
In this section, we explain why maximal reachability is difficult to solve from samples. In the model-free setting, one naturally relies on bootstrap updates such as TD and Q-learning. The difficulty is that, in the undiscounted case, the Bellman equations may admit multiple solutions, whereas reachability probabilities are given by the least fixed-point solutions. 
Consequently, bootstrap methods may converge to a solution that does not equal the reachability probability.

\begin{example}
\label{ex:minimal_drift}
Consider the following MDP.
At state $s$, action $u_1$ keeps the process at $s$, whereas action $u_2$ reaches a target state $z\in Z$ with probability $1/2$ and a non-target terminal state $f\in S\setminus Z$ with probability $1/2$.
Thus $V^\star(s)=1/2$.
The Bellman optimality equation gives
\[
Q^\star(s,u_2)=\tfrac{1}{2},\,
Q^\star(s,u_1)=\max\{Q^\star(s,u_1),Q^\star(s,u_2)\}.
\]
Hence the Bellman optimality equation admits any assignment of the form \(Q^\star(s,u_2)=1/2\) and \(Q^\star(s,u_1)\ge 1/2\) as a solution. The desired reachability values are given by its least solution.
Therefore, a sample-based bootstrap method may not converge to a solution equal to the reachability probability and incorrectly prefer the self\mbox{-}loop action $u_1$, which yields zero reachability.
\end{example}

This phenomenon is not specific to that example.
It arises from recurrent structure in the Markov chain.

\begin{definition}\label{def:BSCC}
A bottom strongly connected component (BSCC) of a Markov chain is a strongly connected component without outgoing transitions.
\end{definition}

On such BSCCs, sample\mbox{-}based bootstrap updates may settle at an arbitrary fixed-point of the Bellman equation.
Under a policy $\pi$, any BSCC $C$ disjoint from $Z$ satisfies $V^\pi(s)=0$ for all $s\in C$, since once the path enters $C$, it cannot reach $Z$.
However, the Bellman equation does not uniquely determine these zero values.
Restricted to $C$, 
\begin{align}
V^\pi|_C=T_C^\pi\, V^\pi|_C,
\label{eq:motivation_bscc}
\end{align}
where $V^\pi|_C\in\mathbb R^{|C|}$ is the restriction of $V^\pi$ to $C$, and $T_C^\pi\in\mathbb R^{|C|\times |C|}$ is the corresponding submatrix of \(T^\pi\).
Since $C$ is a BSCC, every constant vector on $C$ satisfies \eqref{eq:motivation_bscc}.
In Example~\ref{ex:minimal_drift}, under action $u_1$, $\{s\}$ is such a BSCC.
\section{Methodology}
\label{sec:method}
In this section, we develop a model-free policy iteration algorithm for maximal reachability. For each policy, we identify from samples the states responsible for BSCC\mbox{-}induced nonuniqueness, perform undiscounted policy evaluation on the remaining states, and then apply greedy policy improvement. Formal guarantees are deferred to Sec.~\ref{sec:convergence}.

\begin{algorithm}[b]
\footnotesize
\caption{Model-Free Policy Iteration for Reachability}
\label{alg:pdF_pi}
\begin{algorithmic}[1]
\REQUIRE Simulator access to \(s'\sim T(\cdot|s,a)\); target set \(Z\); discount factor \(\gamma\in(0,1)\); clamp threshold \(\varepsilon\); improvement threshold \(\delta\)
\STATE Initialize a deterministic stationary policy \(\pi_0\)
\FOR{\(k=0,1,2,\ldots\)}
    \STATE Estimate \(\widehat V_\gamma^{\pi_k}\) from samples under \(\pi_k\) using \eqref{eq:td_disc_eval_method}
    \STATE Update \(\widehat F^{\pi_k}\) by \eqref{eq:Fhatpi_method}
    \STATE Estimate \(\widehat V_{\widehat F^{\pi_k}}^{\pi_k}\) from samples under \(\pi_k\) using \eqref{eq:td_undisc_eval_method}
    \FORALL{\(s\in S\setminus Z\), \(a\in A(s)\)}
        \STATE Estimate \(\widehat Q_k(s,a)\) from samples by \eqref{eq:Qhat_method}
    \ENDFOR
    \FORALL{\(s\in S\setminus Z\)}
        \STATE Update \(\pi_{k+1}(s)\) according to \eqref{eq:strict_tie_stay_method}
    \ENDFOR
    \STATE \textbf{if} \(\pi_{k+1}=\pi_k\) \textbf{then terminate}
\ENDFOR
\RETURN \(\pi_k\)
\end{algorithmic}
\end{algorithm}

\vspace{-10pt}
\subsection{Model-Free Policy Evaluation for Reachability}
\label{subsec:pe}

Fix a deterministic stationary policy \(\pi\). Our goal is to remove the source of nonuniqueness in the undiscounted Bellman equation before performing policy evaluation. To identify the states responsible for the BSCC-induced nonuniqueness from samples, we introduce the discounted surrogate
\begin{align}
V_\gamma^\pi(s) := \mathbb E_s^\pi[\gamma^{\tau_Z}\mathbf 1\{\tau_Z<\infty\}], \quad \gamma\in(0,1),
\label{eq:Vgpi_method}
\end{align}
where \(\gamma\) is a discount factor.
To isolate this source of nonuniqueness, define the clamp set
\begin{align}
F^\pi
:=
\{s\in S:V_\gamma^\pi(s)=0\}
=
\{s\in S:V^\pi(s)=0\}.
\label{eq:Fpi_method}
\end{align}
Thus \(F^\pi\) contains all BSCCs disjoint from \(Z\) under \(\pi\), together with transient states with \(V^\pi=0\).
Clamping these transient states does not affect the value function on \(S\setminus F^\pi\). 

Here, \(V_\gamma^\pi\) can be estimated from samples. For \(s_t\notin Z\),
\begin{align}
\widehat V_\gamma^\pi(s_t)
\leftarrow
(1-\alpha_t)\widehat V_\gamma^\pi(s_t)
+\alpha_t\,\gamma\,\widehat V_\gamma^\pi(s_{t+1}),
\label{eq:td_disc_eval_method}
\end{align}
where \(\alpha_t\) is a stepsize, \(s_{t+1}\sim T(\cdot\mid s_t,\pi(s_t))\), and \(\widehat V_\gamma^\pi(s)=1\) for \(s\in Z\). 
The estimated clamp set is
\begin{align}
\widehat F^\pi
:=
\{s\in S:\widehat V_\gamma^\pi(s)\le \varepsilon\},
\label{eq:Fhatpi_method}
\end{align} with threshold \(\varepsilon>0\). 
Define the exact and estimated evaluation domains by
$S_0^\pi := S\setminus (Z\cup F^\pi)$,
$\widehat S_0^\pi := S\setminus (Z\cup \widehat F^\pi)$.
We then perform undiscounted policy evaluation on the estimated domain \(\widehat S_0^\pi\). In implementation, \(\widehat F^\pi\) need not be constructed explicitly. 
For \(s\in\widehat S_0^\pi\), consider the clamped undiscounted Bellman equation, which is now well\mbox{-}posed. 
\begin{align}
V_{\widehat F^\pi}^\pi(s)
=
\sum_{s'\in\widehat S_0^\pi} T^\pi(s'|s)\,V_{\widehat F^\pi}^\pi(s')
+
\sum_{z\in Z} T^\pi(z|s).
\label{eq:clamped_bellman_method}
\end{align}
Generally, for any policy \(\pi\) and any clamp set \(F\subseteq S\setminus Z\), \(V_F^\pi\) is the value function under \(\pi\) with boundary value \(0\) on \(F\) and \(1\) on \(Z\). The policy and clamp set indices vary independently.
The corresponding sample recursion is
\begin{align}
\widehat V_{\widehat F^\pi}^\pi(s_t)
\leftarrow
(1-\alpha_t)\widehat V_{\widehat F^\pi}^\pi(s_t)
+\alpha_t\,\widehat V_{\widehat F^\pi}^\pi(s_{t+1}),
\label{eq:td_undisc_eval_method}
\end{align}
for \(s_t\in\widehat S_0^\pi\), with boundary conditions \(\widehat V_{\widehat F^\pi}^\pi(s)=1\) if \(s\in Z\) and \(\widehat V_{\widehat F^\pi}^\pi(s)=0\) if \(s\in \widehat F^\pi\).
When \(\widehat F^\pi=F^\pi\), equivalently \(\widehat S_0^\pi=S_0^\pi\), \(V_{\widehat F^\pi}^\pi\) agrees with the exact reachability probability \(V^\pi\), as shown in Sec.~\ref{sec:convergence}.

\vspace{-10pt}
\subsection{Model-Free Policy Iteration for Reachability}
\label{subsec:pi}

We now introduce a model-free policy iteration algorithm for maximal reachability.
Given policy \(\pi_k\), for each \(s\in S\setminus Z\) and \(a\in A(s)\), we draw \(N_Q\) independent samples and estimate the one-step state-action value by
\begin{align}
\widehat Q_k(s,a)
:=
\frac{1}{N_Q}\sum\nolimits_{i=1}^{N_Q}
\widehat V_{\widehat F^{\pi_k}}^{\pi_k}(s_i').
\label{eq:Qhat_method}
\end{align}

For each \(s\in S\setminus Z\), let \(a_k^\star(s)\in \arg\max_{a\in A(s)} \widehat Q_k(s,a)\). The next policy is updated by
\begin{align}
\pi_{k+1}(s)&=a_k^\star(s)
&&\text{if } \widehat Q_k\!\bigl(s,a_k^\star(s)\bigr)>
\widehat Q_k\!\bigl(s,\pi_k(s)\bigr)+\delta,\nonumber\\
\pi_{k+1}(s)&=\pi_k(s)
&&\text{otherwise},
\label{eq:strict_tie_stay_method}
\end{align}
where \(\delta\ge 0\) is an improvement threshold.
Keeping the current action unless there exists a strict improvement is essential for the policy improvement result in Lemma~\ref{lem:monotone_improvement_exact}.

\section{Convergence Analysis}
\label{sec:convergence}

This section analyzes the convergence of Algorithm~\ref{alg:pdF_pi}. For each fixed policy, policy evaluation recovers the exact reachability value function, and then establishes a convergence guarantee for policy iteration under the following conditions.

\begin{itemize}
    \item \(\textup{(PE1)}\) the discounted TD recursion \eqref{eq:td_disc_eval_method} converges almost surely to \(V_\gamma^\pi\);
    \item \(\textup{(PE2)}\) when exact recovery of the clamp set is claimed, the threshold \(\varepsilon\) satisfies
\(\varepsilon < \min \{\, V_\gamma^\pi(s): V_\gamma^\pi(s)>0 \,\}\);
\item \(\textup{(PE3)}\) for the clamped TD recursion \eqref{eq:td_undisc_eval_method}, the Robbins--Monro stepsize conditions hold and every state in \(\widehat S_0^\pi\) is visited infinitely often.
\end{itemize}
\begin{remark}
\label{rem:eval_assumptions}
Conditions (PE1) and (PE3) are standard for asynchronous tabular TD with Robbins--Monro stepsizes and infinite state visitation~\cite{tsitsiklis1994asynchronous,bertsekas1996neuro}, easily enforced under simulator access. Within (PE2), the choice of \(\varepsilon\) and \(\gamma\) controls the accuracy of clamp set recovery rather than the uniqueness of the clamped Bellman equation~\eqref{eq:clamped_bellman_method}, which only requires \(F^\pi \subseteq \widehat F^\pi\). (PE2) further ensures \(\widehat F^\pi = F^\pi\) for exact recovery. A larger \(\gamma\) widens the valid \(\varepsilon\) region at the cost of slower discounted TD convergence. A quantitative sensitivity analysis is left to future work.
\end{remark}

\vspace{-5pt}
\subsection{Convergence of Sample-Based Policy Evaluation}
\label{subsec:sample_eval}

For policy evaluation,
we first show that the discounted evaluation stage eventually identifies the correct clamp set, and then show that the corresponding clamped evaluation stage converges to the true reachability value \(V^\pi\).

The next lemma shows that 
the estimated clamp set eventually contains and converges exactly to the true clamp set under a mild separation condition.

\begin{lemma}
\label{lem:partition_stabilization}
Fix a policy \(\pi\), and let \(\widehat F_m^\pi:=\{s\in S:\widehat V_{\gamma,m}^\pi(s)\le \varepsilon\}\) denote the estimated clamp set after \(m\) discounted TD updates~\eqref{eq:td_disc_eval_method}.
Under \(\textup{(PE1)}\), \(\widehat F_m^\pi \supseteq F^\pi\) eventually almost surely, where \(F^\pi\) is defined in \eqref{eq:Fpi_method}.
Moreover, under \(\textup{(PE1)}\)--\(\textup{(PE2)}\), \(\widehat F_m^\pi = F^\pi\) eventually almost surely.
\end{lemma}

\begin{proof}
If \(\|\widehat V_{\gamma,m}^\pi - V_\gamma^\pi\|_\infty \le \varepsilon\), then for every \(s\in F^\pi\), since \(V_\gamma^\pi(s)=0\), we have \(\widehat V_{\gamma,m}^\pi(s) \le |\widehat V_{\gamma,m}^\pi(s)-V_\gamma^\pi(s)| \le \varepsilon\), so \(s\in \widehat F_m^\pi\). Hence \(F^\pi\subseteq \widehat F_m^\pi\), and the eventual inclusion follows from the almost-sure convergence of discounted TD.

If in addition \(\min \{\, V_\gamma^\pi(s): V_\gamma^\pi(s)>0 \,\}>\varepsilon\), then almost\mbox{-}sure convergence implies that there exists \(m'\ge m\) such that \(\|\widehat V_{\gamma,m'}^\pi - V_\gamma^\pi\|_\infty < \min \{\, V_\gamma^\pi(s): V_\gamma^\pi(s)>0 \,\}-\varepsilon\).
Therefore, for every \(s\notin F^\pi\),
\begin{align}
\widehat V_{\gamma,m'}^\pi(s)
\ge
V_\gamma^\pi(s)-\|\widehat V_{\gamma,m'}^\pi - V_\gamma^\pi\|_\infty
>
\varepsilon,
\label{eq:notin_Fhat_conv}
\end{align}
so \(s\notin \widehat F_{m'}^\pi\) eventually. Combining this with \(F^\pi\subseteq \widehat F_{m'}^\pi\) yields \(\widehat F_{m'}^\pi = F^\pi\) eventually almost surely.
\end{proof}

Once the estimated clamp set \(\widehat F^\pi\) contains \(F^\pi\), the clamped Bellman equation is uniquely solvable. Furthermore, when \(\widehat F^\pi=F^\pi\), its solution equals \(V^\pi\). 
The following lemma establishes the convergence of policy evaluation.

\begin{lemma}
\label{lem:clamped_td_conv}
Fix a policy \(\pi\).
If \(F^\pi\subseteq \widehat F^\pi\), then the clamped Bellman equation \eqref{eq:clamped_bellman_method} on \(\widehat S_0^\pi\) admits a unique solution \(V_{\widehat F^\pi}^\pi\), and under \textnormal{(PE3)} the recursion \eqref{eq:td_undisc_eval_method} converges almost surely to \(V_{\widehat F^\pi}^\pi\).
Moreover, if \(\widehat F^\pi=F^\pi\), then \(V_{F^\pi}^\pi(s)=V^\pi(s)\) for all \(s\in S_0^\pi\).
\end{lemma}
\begin{proof}
All matrix-vector identities below are taken on \(\widehat S_0^\pi\), with \(V_{\widehat F^\pi}^\pi\) read as its restriction to \(\widehat S_0^\pi\).
Express \eqref{eq:clamped_bellman_method} on \(\widehat S_0^\pi\) in matrix form as
\begin{align}
u = \bar T_{\pi,\widehat F^\pi}u + b_{\pi,\widehat F^\pi},
\label{eq:clamped_matrix_conv}
\end{align}
where \(u\) is the unknown value vector on \(\widehat S_0^\pi\), for \(s,s'\in \widehat S_0^\pi\),
\begin{align}
\bar T_{\pi,\widehat F^\pi}(s,s')&:=T(s'\mid s,\pi(s)), \notag \\
b_{\pi,\widehat F^\pi}(s)&:=\sum\nolimits_{z\in Z} T(z\mid s,\pi(s)). \notag
\end{align}
Since \(F^\pi\subseteq \widehat F^\pi\), every BSCC induced by \(\pi\) that is disjoint from \(Z\) is contained in \(\widehat F^\pi\).
Hence every state in \(\widehat S_0^\pi\) is transient, and therefore \(\bar T_{\pi,\widehat F^\pi}^n \to 0\) as \(n\to\infty\).
Thus the spectral radius \(\rho(\bar T_{\pi,\widehat F^\pi})<1\).
It follows that \(I-\bar T_{\pi,\widehat F^\pi}\) is invertible, so \eqref{eq:clamped_matrix_conv} has the unique solution
\begin{align}
V_{\widehat F^\pi}^\pi
=
(I-\bar T_{\pi,\widehat F^\pi})^{-1} b_{\pi,\widehat F^\pi}.
\end{align}
If \(\widehat F^\pi=F^\pi\), then the clamp set is exactly the set of states with zero reachability probability under \(\pi\).
Hence, by uniqueness of the clamped Bellman equation, for \(s\in S_0^\pi\),
\begin{align}
V_{F^\pi}^\pi(s)=V^\pi(s).
\end{align}
We now analyze the sample recursion \eqref{eq:td_undisc_eval_method}.
This recursion performs tabular TD updates on \(\widehat S_0^\pi\), with boundary values fixed at \(1\) on \(Z\) and \(0\) on \(\widehat F^\pi\).
Hence it is an asynchronous fixed policy tabular TD scheme for the linear system \eqref{eq:clamped_matrix_conv}.
To invoke stochastic-approximation results, it remains to verify the associated ordinary differential equation, the stability of its equilibrium, the boundedness of the iterates, and the regularity of the noise sequence.
First, the one-step TD error can be decomposed into the averaged drift \(b_{\pi,\widehat F^\pi}-(I-\bar T_{\pi,\widehat F^\pi})u\) plus a martingale\mbox{-}difference noise term.
Accordingly, the associated ordinary differential equation is
\begin{align}
\dot u
=
b_{\pi,\widehat F^\pi}-(I-\bar T_{\pi,\widehat F^\pi})u.
\label{eq:ode_conv}
\end{align}
Since \(\rho(\bar T_{\pi,\widehat F^\pi})<1\), every eigenvalue of \(\bar T_{\pi,\widehat F^\pi}-I\) has strictly negative real part.
Hence the equilibrium \(V_{\widehat F^\pi}^\pi\) is globally asymptotically stable.
Second, if initialized in \([0,1]\) with \(0<\alpha_t\le 1\), the iterates remain bounded, since each TD update is a convex combination of the current value and a target in \([0,1]\).
Third, in the finite-state setting, the one-step TD noise is a martingale-difference sequence with bounded second moment.
By \textnormal{(PE3)}, the stepsize \(\alpha_t\) satisfies the Robbins--Monro conditions \(\sum_t \alpha_t=\infty\) and \(\sum_t \alpha_t^2<\infty\), and every state in \(\widehat S_0^\pi\) is visited infinitely often.
Therefore, Prop.~4.6 of~\cite{bertsekas1996neuro} implies that the clamped TD iterates converge almost surely to \(V_{\widehat F^\pi}^\pi\).
\end{proof}

\vspace{-10pt}
\subsection{Convergence of Policy Iteration}
\label{subsec:sample_improve}

Now we present the convergence of the whole policy iteration. Since the classical policy improvement theorem does not apply in the undiscounted reachability setting, we prove monotonicity for a single policy improvement step, and then combine it with the policy evaluation.

\begin{lemma}
\label{lem:monotone_improvement_exact}
Given \(\pi\), let \(\pi^+\) be given by the policy improvement step \eqref{eq:strict_tie_stay_method}.
Then for all \(s\in S\), \(V^{\pi^+}(s)\ge V^\pi(s)\).
If \(\pi^+\neq\pi\), then the inequality is strict at some state.
\end{lemma}

\begin{proof}
We proceed in three steps: set invariance, value function comparison, and strict improvement.

First, we show that the update from \(\pi\) to \(\pi^+\) does not create any new BSCC disjoint from \(Z\) outside \(F^\pi\).
By construction, for every \(s\notin Z\),
\begin{align}
Q^\pi(s,\pi^+(s))\ge V^\pi(s),
\label{eq:pi_plus_ge_vpi}
\end{align}
with strict inequality whenever \(\pi^+(s)\neq\pi(s)\).
Let \(C\) be any BSCC induced by \(\pi^+\) such that \(C\cap Z=\emptyset\), and choose \(s^\star\in C\) maximizing \(V^\pi\) over \(C\).
Since \(C\) is closed, 
\begin{align}
Q^\pi(s^\star,\pi^+(s^\star))
&=
\sum\nolimits_{s'}T(s'|s^\star,\pi^+(s^\star))V^\pi(s') \notag\\
&\le V^\pi(s^\star).
\label{eq:max_state_bound_compact}
\end{align}
Together with \eqref{eq:pi_plus_ge_vpi}, this implies equality in \eqref{eq:max_state_bound_compact}.
Hence every successor of \(s^\star\) under \(\pi^+\) has the same \(V^\pi\)-value as \(s^\star\), and also \(\pi^+(s^\star)=\pi(s^\star)\).
Propagating this argument along paths in \(C\), we conclude that \(V^\pi\) is constant on \(C\) and \(\pi^+(s)=\pi(s)\) for all \(s\in C\).
Thus \(C\) is also a BSCC under \(\pi\). Since \(C\cap Z=\emptyset\), no state in \(C\) can reach \(Z\) under \(\pi\), and hence \(C\subseteq F^\pi\).
Therefore, every BSCC induced by \(\pi^+\) that is disjoint from \(Z\) is contained in \(F^\pi\).

Second, we compare values after clamping \(F^\pi\).
All matrix-vector identities in this step are taken on \(S_0^\pi\), with \(V^\pi\) and \(V_{F^\pi}^{\pi^+}\) read as their restrictions to \(S_0^\pi\).
Since every BSCC of \(\pi^+\) disjoint from \(Z\) lies in \(F^\pi\), the subchain on \(S_0^\pi\) is transient under \(\pi^+\).
Therefore the clamped Bellman equation under \(\pi^+\) has a unique solution \(V_{F^\pi}^{\pi^+}\), satisfying
\begin{align}
V_{F^\pi}^{\pi^+}
=
\bar T_{\pi^+,F^\pi}V_{F^\pi}^{\pi^+}+b_{\pi^+,F^\pi}.
\label{eq:clamped_bellman_pi_plus_compact}
\end{align}
Equivalently, \((I-\bar T_{\pi^+,F^\pi})V_{F^\pi}^{\pi^+}=b_{\pi^+,F^\pi}\).
Moreover, by the boundary conditions, for every \(s\in S_0^\pi\),
\begin{align}
\bigl(\bar T_{\pi^+,F^\pi}V^\pi+b_{\pi^+,F^\pi}\bigr)(s)
=
Q^\pi(s,\pi^+(s))
\ge V^\pi(s).
\label{eq:operator_on_vpi}
\end{align}
Therefore,
\begin{align}
&(I-\bar T_{\pi^+,F^\pi})(V_{F^\pi}^{\pi^+}-V^\pi) \notag\\
&=(I-\bar T_{\pi^+,F^\pi})V_{F^\pi}^{\pi^+}-(I-\bar T_{\pi^+,F^\pi})V^\pi \notag\\
&=b_{\pi^+,F^\pi}-\bigl(V^\pi-\bar T_{\pi^+,F^\pi}V^\pi\bigr) \notag\\
&=\bar T_{\pi^+,F^\pi}V^\pi+b_{\pi^+,F^\pi}-V^\pi
\ge 0.
\label{eq:linear_compare_compact}
\end{align}
Because \(\rho(\bar T_{\pi^+,F^\pi})<1\), the inverse \((I-\bar T_{\pi^+,F^\pi})^{-1}\) exists and is entrywise nonnegative, so \(V_{F^\pi}^{\pi^+}\ge V^\pi\) on \(S_0^\pi\).
Also, by construction of the clamped value function,
\begin{align}
V_{F^\pi}^{\pi^+}(s)
=
\mathbb P_s^{\pi^+}(\tau_Z<\tau_{F^\pi}) \notag
\le
\mathbb P_s^{\pi^+}(\tau_Z<\infty)
=
V^{\pi^+}(s),
\label{eq:clamped_below_true_value}
\end{align}
for every \(s\in S_0^\pi\).
Combining this with \(V_{F^\pi}^{\pi^+}\ge V^\pi\), we obtain \(V^{\pi^+}\ge V^\pi\) on \(S_0^\pi\). The same inequality is trivial on \(F^\pi\) and \(Z\).
Therefore \(V^{\pi^+}\ge V^\pi\) on all of \(S\).

Third, we prove strict improvement when \(\pi^+\neq\pi\).
If some changed state lies in \(S_0^\pi\), then by \eqref{eq:pi_plus_ge_vpi}, the vector
\[
r:=\bar T_{\pi^+,F^\pi}V^\pi+b_{\pi^+,F^\pi}-V^\pi
\]
is nonnegative and strictly positive in at least one component.
By \eqref{eq:linear_compare_compact}, we have 
$
V_{F^\pi}^{\pi^+}-V^\pi=(I-\bar T_{\pi^+,F^\pi})^{-1}r
$.
Since \(\rho(\bar T_{\pi^+,F^\pi})<1\), componentwise we have
\[
(I-\bar T_{\pi^+,F^\pi})^{-1}
=
\sum\nolimits_{n=0}^{\infty}\bar T_{\pi^+,F^\pi}^n
\ge I.
\]
Hence, if \(r_j>0\) for some component \(j\), then
\[
\bigl(V_{F^\pi}^{\pi^+}-V^\pi\bigr)_j
=
\bigl((I-\bar T_{\pi^+,F^\pi})^{-1}r\bigr)_j
\ge r_j>0.
\]
Therefore \(V_{F^\pi}^{\pi^+}>V^\pi\) at some state, and thus \(V^{\pi^+}>V^\pi\) at some state as well.

Otherwise, every changed state lies in \(F^\pi\).
Then for some \(s\in F^\pi\), we have \(Q^\pi(s,\pi^+(s))>V^\pi(s)=0\).
Using the already established monotonicity \(V^{\pi^+}\ge V^\pi\), we obtain
\begin{align}
V^{\pi^+}(s)
&=
\sum\nolimits_{s'}T(s'|s,\pi^+(s))V^{\pi^+}(s') \notag\\
&\ge
\sum\nolimits_{s'}T(s'|s,\pi^+(s))V^\pi(s') \notag\\
&=
Q^\pi(s,\pi^+(s))
>
0.
\label{eq:strict_improvement_on_F}
\end{align}
Thus \(V^{\pi^+}>V^\pi\) at some state in either case.
\end{proof}

Finally, we combine the policy evaluation with the monotonic policy improvement to show that, when the action values used for policy improvement are evaluated exactly, Algorithm~\ref{alg:pdF_pi} terminates after finitely many policy improvements at an optimal policy with probability one.

\begin{theorem}
\label{thm:sample_based_pi_conv}
Under conditions \(\textup{(PE1)}\)--\(\textup{(PE3)}\), if policy improvement is performed with \(\delta=0\) and
\(\widehat Q_k(s,a)=Q^{\pi_k}(s,a)\) for every iteration \(k\) and all relevant \(s,a\),
then Algorithm~\ref{alg:pdF_pi} terminates after finitely many policy improvements at an optimal policy with probability \(1\).
\end{theorem}

\begin{proof}
Let \(\Pi_{\mathrm{DS}}\) denote the finite set of deterministic stationary policies.
For each \(\pi\in\Pi_{\mathrm{DS}}\), under conditions \(\textup{(PE1)}\)--\(\textup{(PE3)}\), the policy evaluation stage returns the exact value \(V^\pi\) almost surely when run to convergence.
Since \(\Pi_{\mathrm{DS}}\) is finite, the intersection of these almost-sure events over all \(\pi\in\Pi_{\mathrm{DS}}\) still has probability one.
Fix a sample path in this event.
On this sample path, because \(\delta=0\) and the one-step values are evaluated exactly, each policy improvement step coincides with the exact update studied in Lemma~\ref{lem:monotone_improvement_exact}.

Whenever \(\pi_{k+1}\neq\pi_k\), Lemma~\ref{lem:monotone_improvement_exact} yields $V^{\pi_{k+1}}(s)\ge V^{\pi_k}(s)$ for all $s\in S$ 
with strict inequality at some state.
Hence no policy can repeat along this sample path.
Since \(\Pi_{\mathrm{DS}}\) is finite, only finitely many policy improvements that change the policy are possible.
Therefore, the policy iteration procedure terminates after finitely many policy improvements at some policy \(\bar\pi\).

For this terminal policy \(\bar\pi\), no strict-improvement action exists.
Hence, for every \(s\notin Z\),
\[
V^{\bar\pi}(s)
=
Q^{\bar\pi}(s,\bar\pi(s))
=
\max\nolimits_{a\in A(s)} Q^{\bar\pi}(s,a).
\]
Together with the boundary condition \(V^{\bar\pi}(z)=1\) for \(z\in Z\), this shows that \(V^{\bar\pi}\) satisfies the Bellman optimality equation \eqref{eq:BOE}.
Since \(V^\star\) is the least fixed-point of \eqref{eq:BOE} and every policy value satisfies \(V^{\bar\pi}\le V^\star\), we obtain \(V^\star\le V^{\bar\pi}\le V^\star\).
Therefore \(V^{\bar\pi}=V^\star\), and \(\bar\pi\) is optimal.
\end{proof}

\begin{remark}
\label{rem:Q_concentration}
The exact-$\widehat Q$ assumption in Theorem~\ref{thm:sample_based_pi_conv} is a simplification keeping the statement focused on the resolution of nonuniqueness. For finite $N_Q$, $\widehat Q_k$ in \eqref{eq:Qhat_method} concentrates around $Q^{\pi_k}$ at a Hoeffding rate. With a strict-improvement threshold $\delta>0$ and the Hoeffding error of $\widehat Q_k$ kept below $\delta/2$, Algorithm~\ref{alg:pdF_pi} terminates with high probability at a policy $\bar\pi$ such that $\max_{a} Q^{\bar\pi}(s,a) - V^{\bar\pi}(s) \le 2\delta$.
\end{remark}
\section{Case Study}\label{sec:case_study}
We validate the proposed method on a deliberately misleading stochastic grid world.
The \(5\times 4\) grid world has initial state \((0,1)\), intermediate goal \(A=(0,3)\), final goal \(B=(3,1)\), and absorbing trap states \(T\) at all cells \((x,y)\) with \(x\in\{1,2,3\}\) and \(y\in\{0,2,4\}\). The task is specified by the sc\mbox{-}LTL formula
\(\varphi=\Diamond(A\wedge\Diamond B)\) 
which requires the path to first reach \(A\) and then eventually reach \(B\). The traps block all vertical movement between rows \(y=1\) and \(y=3\) except through column \(x=0\). Intuitively, the agent must move upward along the left corridor to \(A\), then return and cross to the right corridor to reach \(B\). At each nonterminal state, the commanded action succeeds with probability \(0.8\), and the remaining probability is split equally between the two adjacent directions. Translating \(\varphi\) into a DFA yields a three-state automaton. The resulting product MDP has \(60\) states, where \(s_0\) denotes its initial state.

This instance is designed to contain two difficulties for sample-based planning. First, the discounted surrogate is deliberately misleading on this example. For \(\gamma=0.95\), a suboptimal policy \(\pi'\) has \(V_\gamma^{\pi'}(s_0)\approx 0.404\) but \(V^{\pi'}(s_0)=0.675\), whereas the optimal policy \(\pi^\star\) has \(V^{\pi^\star}(s_0)\approx 0.780\) and a much smaller \(V_\gamma^{\pi^\star}(s_0)\approx 0.017\). Thus, a policy that appears favorable under the discounted surrogate can be strongly suboptimal for the original sc\mbox{-}LTL objective. Second, the absorbing trap states induce BSCCs disjoint from \(Z\), which cause nonuniqueness in the undiscounted Bellman equation. This case study therefore tests whether the proposed method can still recover the correct optimal policy for the sc\mbox{-}LTL objective despite both obstacles.

The experiments were conducted using Python on a Windows 11 machine equipped with an Intel i9-14900K processor.
We first solve the planning problem exactly by value iteration to obtain the ground truth, and then run the proposed policy iteration using samples only. We use \(\gamma=0.95\), \(\varepsilon=1e{-}3\), \(8e3\) discounted TD updates, \(3e4\) clamped TD updates, \(1e4\) samples in each policy improvement step, constant stepsize \(\alpha=0.05\), and threshold \(1e{-}4\).

\begin{figure}[t]
    \centering
    \includegraphics[width=\linewidth]{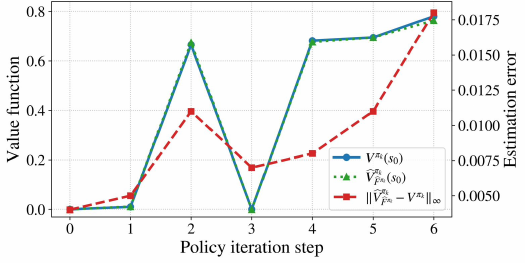}
    \vspace{-20pt}
    \caption{Evolution of the exact value, the sample-based estimate, and the estimation error over policy iteration steps in the grid world. The estimate remains accurate, and the policy iteration reaches the optimal value.}
    \label{fig:case_outerloop}
    \vspace{-15pt}
\end{figure}

These numerical results are consistent with the theory in Sec.~\ref{sec:convergence}. As shown in Fig.~\ref{fig:case_outerloop}, the algorithm terminates after \(7\) policy improvements at a final policy with \(V^{\pi_K}(s_0)=V^\star(s_0)=0.780\), indicating empirical recovery of the optimal policy and consistency with Theorem~\ref{thm:sample_based_pi_conv}. Moreover, the sample-based clamped estimate \(\widehat V_{\widehat F^{\pi_k}}^{\pi_k}(s_0)\) closely tracks the exact value \(V^{\pi_k}(s_0)\) throughout the iterations. The clamped policy evaluation is accurate, with \(\max_k \|\widehat V_{\widehat F^{\pi_k}}^{\pi_k}-V^{\pi_k}\|_\infty = 1.8\times 10^{-2}\), supporting Lemma~\ref{lem:clamped_td_conv}. In addition, the final estimated clamp set satisfies \(\widehat F^{\pi_K}=F^{\pi_K}\) with \(|\widehat F^{\pi_K}|=18\), consistent with Lemma~\ref{lem:partition_stabilization}. Finally, although the value sequence is not strictly monotone at every policy improvement step because policy evaluation is sample-based, most updates are still improving, in qualitative agreement with Lemma~\ref{lem:monotone_improvement_exact}.
\section{Conclusion} 
This work studied model-free policy iteration for sc\mbox{-}LTL planning via maximal reachability. The method resolves the nonuniqueness of the undiscounted Bellman equation by identifying relevant BSCCs through a discounted surrogate, followed by well-posed undiscounted evaluation and greedy improvement. We proved almost-sure convergence of policy evaluation and, under exact one-step improvement, finite termination at an optimal policy.

\bibliographystyle{IEEEtran}

\end{document}